\documentclass[runningheads]{llncs}
\usepackage{amsmath}
\usepackage{amssymb}
\usepackage{hyperref}
\usepackage{tikz}
\usepackage[T1]{fontenc}
\usepackage{lmodern}
\usepackage{slantsc}

\rmfamily % To load Latin Modern Roman and enable the following NFSS declarations.
\DeclareFontShape{T1}{lmr}{b}{sc}{<->ssub*cmr/bx/sc}{}
\DeclareFontShape{T1}{lmr}{bx}{sc}{<->ssub*cmr/bx/sc}{}

\usetikzlibrary{arrows,automata,shapes.geometric,positioning,
  decorations.pathmorphing,decorations.pathreplacing,
  calc,graphs}

\tikzset{
  ->, % arrow
  >=stealth', % arrow head
  auto, % label positioning
  node distance=1cm,
  minimum height=0.3cm,
  initial text={}, % start label
}
\tikzset{every state/.style={minimum size=0cm}}
\tikzset{every node/.append style={font=\footnotesize}}
\tikzset{elliptic state/.style={draw,ellipse,inner sep=2}}
\tikzset{SharpSquiggly/.style={decorate,
  decoration={zigzag, segment length=4, amplitude=0.9}}}

\allowdisplaybreaks

\title{%
  Construction of a DFA for Computing Grundy Numbers
  in the Successful Derivation Games on Right-Linear Grammars%
}
\titlerunning{%
  On the Successful Derivation Games on Right-Linear Grammars%
}
\author{Yoshiaki Takata\inst{1} \and Yusuke Inoue\inst{2} \and
  Hiroyuki Seki\inst{3}}
\institute{
  School of Informatics, Kochi University of Technology \\
  \email{takata.yoshiaki@kochi-tech.ac.jp}
  \and
  Graduate School of Informatics, Nagoya University  \\
  \and
  School of Social Informatics, Mukogawa Women's University \\
}

\newcommand{\Nat}{\mathbb{N}}
\newcommand{\powset}[1]{\mathcal{P}({#1})}
\newcommand{\finpowset}[1]{\mathcal{P}_{\mathrm{fin}}({#1})}

\newcommand{\calG}{\mathcal{G}}
\newcommand{\Grd}[1]{\calG_{#1}}
\newcommand{\mex}{\mathop{\mathrm{mex}}}

\newcommand{\LQ}{\backslash}
\newcommand{\RQ}{\slash}

\newcommand{\posTo}{\rightsquigarrow}
\newcommand{\Deriv}{\Rightarrow}

\newcommand{\Sign}[1]{\mathop{\mathrm{sign}}{#1}}

\newcommand{\Rev}[1]{{#1}^{\mathcal{R}}}
\newcommand{\E}{\mathsf{E}}
\newcommand{\Rn}{R_{\mathsf{N}}}
\newcommand{\Ru}{R_{\mathsf{U}}}
\newcommand{\Rt}{R_{\mathsf{T}}}
\newcommand{\GammaN}{\Gamma_{\mathsf{N}}}
\newcommand{\ALL}{\mathsf{ALL}}

\renewcommand{\t}[1]{\mathtt{#1}}
\renewcommand{\setminus}{-}

\begin{document}
\maketitle

\begin{abstract}
Inoue et al.\ have introduced
the \emph{successful derivation game} (SDG) on context-free grammars (CFGs),
which is a generalization of classic heap-based games
including subtraction games and Keyles,
and shown that
the least upper bound of the Grundy numbers in
the SDG on a given CFG $G$
is undecidable in general even when we restrict $G$ to be a linear CFG\@.
This paper shows that
for the SDG on a right-linear grammar (RLG),
we can construct a DFA for computing the Grundy number of
a given position.
In other words,
for the SDG on an RLG,
the set of positions with a given Grundy number~$c$
is regular.
As a corollary, the least upper bound of the Grundy numbers
in the SDG on a given RLG is decidable.
We also investigate the complexity of
computing the least upper bound of the Grundy numbers
in the SDG on a given RLG,
and it is shown to be PSPACE-complete.
\end{abstract}

%%% ----- Introduction -----

\section{Introduction}

\emph{Combinatorial games} are
``two-player games with no hidden information and
no chance elements''\cite{S13}
and the theory of combinatorial games concerns with
the mathematical structure of those games.
Among them,
\emph{impartial}, \emph{loop-free}
games under the \emph{normal-play} convention are
particularly well-studied,
where impartial means that both players have the same legal moves,
loop-free means that every play is finite,
and normal play means that the player who makes the last move wins.
Examples of such games include classical heap-based games
played on heaps of stones,
such as \emph{subtraction games} including \emph{Nim}.
One of the important aspects in
the study of combinatorial games is
the notion of \emph{Grundy numbers}~\cite{G39,S35}.
A Grundy number is a nonnegative integer assigned to
each position and indicates the position is winning or losing.

Inoue et al.~\cite{IKIK26} have introduced
the \emph{successful derivation game} (SDG) on context-free grammars (CFGs),
which can be seen as a generalization of a subclass of heap-based games
including subtraction games and Keyles.
%The SDG starts with a given target word,
%and players alternately rewrite a word according to a given CFG\@.
They have shown that
the least upper bound of the Grundy numbers in
the SDG on a given CFG $G$
is undecidable in general even when we restrict $G$ to be a linear CFG\@.
%Moreover, it is unknown whether or not
%the set of positions with Grundy number zero
%(called $\mathcal{P}$-positions)
%%(resp.\ the set of positions with a nonzero Grundy number
%% called $\mathcal{N}$-positions)
%is context-free in general.
%
This paper shows that for the SDG on a right-linear grammar (RLG),
we can construct a DFA for computing the Grundy number of
a given position.
In other words,
for the SDG on an RLG,
the set of positions with a given Grundy number~$c$
is regular.
By utilizing the DFA for computing the Grundy number,
we show that the least upper bound of the Grundy numbers
in the SDG on a given RLG is decidable.
We also investigate the complexity of
computing the least upper bound of the Grundy numbers
in the SDG on a given RLG,
and it is shown to be PSPACE-complete.

The outline of this paper is as follows.
In Section~\ref{sec:definitions},
we define the successful derivation game and
related concepts.
In Section~\ref{sec:RLGtoDFA},
we give an algorithm to construct
a DFA representing the set of positions with
a given Grundy number.
The correctness of the algorithm is shown
in Section~\ref{sec:correctness}.
In Section~\ref{sec:complexity},
we show the PSPACE-completeness of computing
the least upper bound of the Grundy numbers
in the SDG on a given RLG\@.
In Section~\ref{sec:examples},
we show two examples of SDGs on RLGs
and examine the resultant DFAs.
%The first one represents
%the subtraction game with colored stones,
%which has been introduced in~\cite{IKIK26}.
%The second one
%represents a game similar to the first one
%and is an example with a unit rule.
Finally, we conclude this paper in Section~\ref{sec:conclusion}.

%%% ----- Definitions -----

\section{Definitions}
\label{sec:definitions}

Let $\Nat=\{0,1,2,\ldots\}$ be the set of natural numbers
including zero.
%Let $\Bool=\{0,1\}$.
For a set~$S$,
let $\powset{S}$ be the power set of~$S$ and
$\finpowset{S}$ be the set of all finite subsets of~$S$.
For a finite set $S$, let $|S|$ denote the cardinality of~$S$.
For a finite word $w$, let $|w|$ denote the length of~$w$.

%A set of words is called a \emph{language}.
For a word $x$, let $\Rev{x}$ denote the \emph{reverse} of~$x$;
i.e., if $x=a_1a_2\ldots a_k$, then $\Rev{x}=a_k\ldots a_2a_1$.
For a language $L$, let $\Rev{L} = \{\Rev{x}\mid x\in L\}$.
For a word $x$ and a suffix $y$ of~$x$,
let $x\RQ y$ denote the right quotient of $x$ with respect to~$y$;
i.e., $x\RQ y$ equals the word $z$ such that $x = zy$.
%For a language $L$ and a word $y$,
%$L\RQ y$ denotes the right quotient of $L$ with respect to $y$;
%i.e.,
%$L\RQ y = \{z\mid zy\in L\}$.
%
For a word $y$ and a prefix $x$ of~$y$,
$x\LQ y$ denotes the left quotient of $y$ with respect to~$x$;
i.e., $x\LQ y$ equals the word $z$ such that $y = xz$.
%For a language $L$ and a word $x$,
%%$xL$ denotes the concatenation of $x$ and $L$ and
%$x\LQ L$ denotes the left quotient of $L$ with respect to $x$;
%i.e.,
%%$xL = \{xy\mid y\in L\}$ and
%$x\LQ L = \{y\mid xy\in L\}$.
%$\SufNe(L)$ denotes the set of all non-empty suffixes of every word in~$L$.

For positions $Q, Q'$ of a combinatorial game,
we write $Q\posTo Q'$
if there exists a legal move from $Q$ to~$Q'$.
The \emph{game graph} of a game is the directed graph where
the nodes are the positions and each edge represents
a legal move.
%The \emph{game tree} of a game rooted at a position $Q$
%is the labeled rooted tree where each node is labeled with a position,
%in particular the root is labeled with~$Q$, and
%for every node $v$, the labels of $v$'s children are the successor positions
%of $v$'s label.

%%%

\begin{definition}[context-free grammar]
  A \emph{context-free grammar} (or \emph{CFG}) is
  a 4-tuple $G=(\Gamma,\Sigma,R,I)$ where
  \begin{itemize}
  \item $\Gamma$ is a finite set of nonterminal symbols,
  \item $\Sigma$ is a finite set of terminal symbols,
  \item $I\in\Gamma$ is the initial symbol, and
  \item $R\subseteq \Gamma\times(\Gamma\cup\Sigma)^+$
    is a finite set of rules.
  \end{itemize}
\end{definition}

A rule $A\to\beta$ is called a \emph{terminal rule} if $\beta\in\Sigma^+$.
A rule $A\to B$ is called a \emph{unit rule} if $B\in\Gamma$.
In this paper, we consider only CFGs
\emph{without a cycle of unit rules};
i.e.,
if CFG $G$ has a unit rule ${A\to B} \in R$,
then there is no derivation of the form $B \Deriv^* A$.
This condition is imposed to ensure that
derivations are of finite length.
A derivation $\alpha\Deriv^*\beta$ with
$\alpha,\beta\in(\Gamma\cup\Sigma)^+$ is defined in the usual way.
%A derivation $\alpha\Deriv^*\beta$ is \emph{successful}
%if $\beta\in\Sigma^+$ and is \emph{incomplete} otherwise.

A CFG $G$ is \emph{linear}
if the right-hand side of each rule of $G$ contains
at most one nonterminal.
A linear CFG $G$ is \emph{right-linear}
if the right-hand side of each rule is an element of
$\Sigma^+\cup\Sigma^*\Gamma$.
% every nonterminal appears only at the end of a right-hand side.
%A \emph{regular} grammar is a right-linear grammar
%where the right-hand side of each rule is an element of
%$\Sigma\cup\Gamma\cup\Sigma\Gamma$.

For each nonterminal $A\in\Gamma$,
let $L(G,A)$ be the language defined as
\begin{equation*}
  L(G,A)=\{w\in\Sigma^+ \mid
    \text{there is a derivation of the form } A\Deriv^* w\}.
\end{equation*}
The language generated by $G = (\Gamma,\Sigma,R,I)$ is defined
as $L(G)=L(G,I)$.

\begin{definition}[successful derivation game\,{\mdseries\cite{IKIK26}}]
  For a CFG $G=(\Gamma,\Sigma,R,I)$ and a word $w\in L(G)$,
  the \emph{successful derivation game} $(G,w)$ is
  the two-player combinatorial game defined as follows.
  \begin{itemize}
  \item A \emph{position} is a finite multiset of the form
    \begin{equation*}
      \{ (A_1,w_1),(A_2,w_2),\ldots,(A_{\ell},w_{\ell}) \}
    \end{equation*}
    where $(A_i,w_i)\in \Gamma\times\Sigma^+$ and
    $w_i\in L(G,A_i)$ for each $1\le i\le\ell$.
  \item A legal move from a position $Q$ is defined as follows.
    First, the current player chooses an element $(A,w)$ of $Q$.
    Next, the player chooses a sequence
    $(A_1,w_1),\ldots,(A_{\ell},w_{\ell})\in \Gamma\times\Sigma^+$
    with $\ell\ge 0$ such that
    \begin{itemize}
    \item $w_i\in L(G,A_i)$ for each $i$, and
    \item $A\to u_0 A_1 u_1 \ldots u_{\ell-1} A_{\ell} u_{\ell} \in R$
      where $w=u_0 w_1 u_1 \ldots u_{\ell-1} w_{\ell} u_{\ell}$.
    \end{itemize}
    Finally, the position $Q$ is updated to
    \begin{equation*}
      Q' = (Q\setminus\{(A,w)\})\cup\{ (A_1,w_1),\ldots,(A_{\ell},w_{\ell}) \}.
    \end{equation*}
  \item The game starts from the initial position $Q_0=\{(I,w)\}$,
    and the players alternate making legal moves $Q_i\posTo Q_{i+1}$
    for each $i\ge 0$.
    The player who is unable to move loses.
  \end{itemize}
\end{definition}

%Note that a successful derivation game is a \emph{impartial}, \emph{loop-free}
%game under the normal-play convention,
%where impartial means that both players have the same legal moves,
%loop-free means that every play is finite,
%and normal play means that the player who makes the last move wins.

%%%

%\bigskip

Let $\mex:\finpowset{\Nat}\to\Nat$ be the function
defined as $\mex(X)=\min(\Nat\setminus X)$,
which means ``the minimum excluded value.''

\begin{definition}[Grundy number]
  \label{def:Grundynum}
  For a position $Q$ of an impartial, loop-free combinatorial game~$G$,
  the \emph{Grundy number} $\Grd{G}(Q)\in\Nat$ is defined as
  \begin{equation*}
    \Grd{G}(Q)=\mex(\{ \Grd{G}(Q') \mid Q\posTo Q' \}).
  \end{equation*}
\end{definition}

Note that if $Q$ has no successive position, then
$\Grd{G}(Q)=0$ by definition.

It is well known that if $G$ is a normal-play game, then
$\Grd{G}(Q)\ne 0$
iff the next player has a winning strategy at~$Q$
(Sprague-Grundy theorem)~\cite{G39,S35}.

Since we assume that a given CFG $G$ has no cycle of unit rules,
the SDG on $G$ is loop-free
and we can define the Grundy numbers of positions in the game
as the same as Definition~\ref{def:Grundynum}.

%%% ----- Example: CSub -----

\begin{figure}[t] \centering
  \begin{tikzpicture}[->, thick, scale=0.9, transform shape]
  \tikzstyle{every node}=[font=\small]
  \foreach \w [count=\i from 0] in {aa, ba, ab, bb}
  {
    \node (Aa\w) at (0, -1.4 * \i) {$(A,\t{a\w})$};
    \node [below=0.7 of Aa\w.center, anchor=center] (Ab\w) {$(A,\t{b\w})$};
  }
  \foreach \w [count=\i from 0] in {ba,aa,ab,bb}
  {
    \node[below right=0.35 and 2.8 of Aa\w.center, anchor=center] (A\w) {$(A,\t{\w})$};
    \draw (Aa\w) -- (A\w);
    \draw (Ab\w) -- (A\w);
  }
  \foreach \w [count=\i from 0] in {a,b}
  {
    \node[below right=0.7 and 2.8 of Aa\w.center, anchor=center] (A\w) {$(A,\t{\w})$};
    \draw (Aa\w) -- (A\w);
    \draw (Ab\w) -- (A\w);
  }
  \node[below right=1.4 and 2.8 of Aa.center, anchor=center] (F) {$\varnothing$};
  \draw (Aa) to [bend left=10] (F);
  \draw (Ab) to [bend right=10] (F);
  \draw (Aaaa) to [bend left=23] (Aa);
  \draw (Aaab) to [bend left=23] (Ab);
  \draw (Aaa) to [bend right=18] (F);
  \tikzstyle{every node}=[font=\footnotesize, text=red]
  \foreach \w / \g in {aaa/0, baa/0, aba/1, bba/1,
                       aab/2, bab/1, abb/1, bbb/1}
  {
    \node[below=0.2 of A\w.north west, anchor=center] {$\g$};
  }
  \foreach \w / \g in {aa/2, ba/0, ab/0, bb/0, a/1, b/1}
  {
    \node[below left=0.1 and 0.2 of A\w.south, anchor=center] {$\g$};
  }
  \node[above=0.1 of F.north, anchor=center] {$0$};
  \end{tikzpicture}
  \caption{A partial game graph of the successful derivation game
    in Example~\ref{example:pre-csub}.}
  \label{fig:graph-csub}
\end{figure}
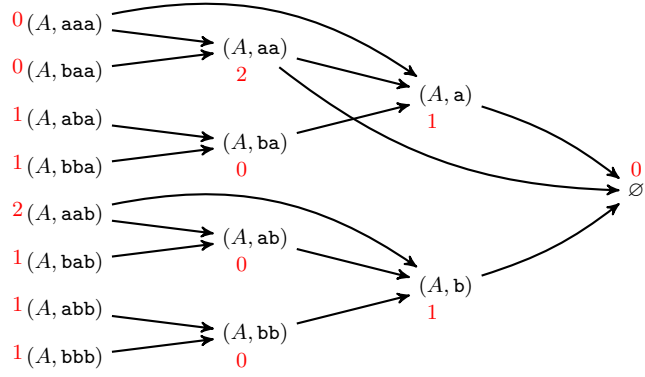%

\begin{example}[Subtraction game with colored stones]
\label{example:pre-csub}
  Consider an RLG
  $G=(\{A\},\allowbreak \Sigma,\allowbreak R,A)$ where $\Sigma=\{\t{a},\t{b}\}$ and
  $R=\{\,
  A \to \t{a}A \mid \t{aa}A \mid \t{b}A \mid \t{a}
  \mid \t{aa} \mid \t{b} \,\}$,
  which equals the grammar $\mathbf{CSub}$
  described in Examples~4 and 7 in~\cite{IKIK26}.
  The SDG on $G$ represents a game in which
  each player can remove one $\t{a}$, two $\t{a}$'s, or one $\t{b}$
  from left to right.

  Figure~\ref{fig:graph-csub} shows a partial game graph
  of the SDG on~$G$.
  In the graph, each singleton position $\{(A,w)\}$
  is written as $(A,w)$ for simplicity.
  The graph contains every position $(A,w)$ such that $|w|\le 3$.
  The graph also shows the Grundy number of each position.
  For example, $\Grd{G}(\{(A,\t{a})\})=1$
  since $\{(A,\t{a})\}$ has only one successor position $\emptyset$
  whose Grundy number is zero,
  and $\Grd{G}(\{(A,\t{aa})\})=2$
  since $\{(A,\t{aa})\}$ has two successor positions
  $\{(A,\t{a})\}$ and $\emptyset$
  whose Grundy numbers are one and zero, respectively.
\end{example}

%%%

%\bigskip

In the following, we consider Grundy numbers of positions of
a successful derivation game
on a CFG $G=(\Gamma,\Sigma,R,I)$.
For a position $Q$ of a successful derivation game $(G,w)$
on $G$,
let $\Grd{G}(Q)$ denote the Grundy number of $Q$,
because $w$ is irrelevant to the value of $Q$'s Grundy number.

For each nonterminal $A\in\Gamma$,
define the function $\Grd{G,A}:\Sigma^+\to \Nat\cup\{\bot\}$ as
\begin{equation*}
  \Grd{G,A}(w)=\begin{cases}
    \Grd{G}(\{(A,w)\}) & \text{if } w\in L(G,A) \\
    \bot & \text{otherwise}
  \end{cases}
\end{equation*}
for each $w\in\Sigma^+$.

Inoue et al.\ have shown the following propositions~\cite{IKIK26}.

\begin{proposition}[computability of $\Grd{G,A}$]
  \label{prop:computable}
  For a CFG $G=(\Gamma,\Sigma,R,I)$,
  a nonterminal $A\in\Gamma$, and a word $w\in\Sigma^+$,
  $\Grd{G,A}(w)$ can be computed in polynomial time
  with respect to the length of $w$.
\end{proposition}
\begin{proposition}[boundedness of $\Grd{G,A}$]
  \label{prop:mA}
  For a linear CFG $G=(\Gamma,\Sigma,R,I)$ and
  a nonterminal $A\in\Gamma$,
  there is an upper bound
  $m_A\in\Nat$ such that $\Grd{G,A}(w)\le m_A$
  for each $w\in L(G,A)$.
\end{proposition}
\begin{proposition}[undecidability of $\max\Grd{G,A}$]
  \label{prop:undecidable}
  The problem to compute
  the least upper bound of Grundy numbers in the SDG,
  defined below, is undecidable even if the input CFG $G$ is
  restricted to be linear.
  \begin{itemize}
  \item[] \hspace{-1em}%
    \textsc{\textsl{Max-Grundy-Number-in-SDG}}
  \item Instance: a CFG $G=(\Gamma,\Sigma,R,I)$,
    a nonterminal $A\in\Gamma$, and $k\in\Nat$.
  \item Question: $\max_{w\in L(G,A)}\bigl(\Grd{G,A}(w)\bigr) = k$ ?
  \end{itemize}
\end{proposition}

%%%

In this paper, we show that
the problem \textsc{Max-Grundy-Number-in-SDG}
described in Proposition~\ref{prop:undecidable}
becomes decidable if we restrict the input CFG $G$ to be
right-linear.
More generally, we show that we can construct a DFA
representing the set of positions with a given Grundy number~$c$,
and the above problem can be reduced to
the emptiness problem for these DFAs.

%%% ----- Construction of a DFA -----

\section{DFA Representing the Set of Positions
  With a Given Grundy Number}
\label{sec:RLGtoDFA}

From here, we concentrate on
the SDG on a given right-linear grammar (RLG)~$G$.
In this section, we describe an algorithm for constructing
a DFA $M_{A,i}$ for each pair
of a nonterminal $A$ of $G$ and a number $i\le m_A$
($m_A$ is an upper bound of $\Grd{G,A}(w)$ described in
Proposition~\ref{prop:mA})
that satisfies
\begin{equation*}
  \Rev{w}\in L(M_{A,i}) \iff \Grd{G,A}(w)=i.
\end{equation*}
We can compute
$\max_{w\in L(G,A)}\Grd{G,A}(w)$
by checking the emptiness of $L(M_{A,i})$ for
each $i\le m_A$,
because
\begin{equation*}
  \max_{w\in L(G,A)}\Grd{G,A}(w)=m \iff
  L(M_{A,m})\ne\emptyset \text{ and }
  \forall i > m.\
  L(M_{A,i})=\emptyset.
\end{equation*}
Every $M_{A,i}$ shares the same
state set $Q$, start state $q_0$,
and transition function~$\delta$,
differing only in the set of final states.

Before showing the algorithm for constructing $M_{A,i}$,
we describe a concrete method to define an upper bound
$m_A$ of~$\Grd{G,A}(w)$.

%%%

\subsection{Defining an Upper Bound $m_A$}
\label{sec:m_A}
Let a given RLG be $G=(\Gamma,\Sigma,R,I)$.
For each $A\in\Gamma$ and $a\in\Sigma$,
define the following sets:
\begin{align*}
  T_{A,a} &= \{ x\in\Sigma^* \mid A\to ax\in R\},\\
  N_{A,a} &= \{ xB\in \Sigma^*\Gamma \mid A\to axB\in R\},\\
  U_{A} &= \{ C\in\Gamma \mid A\to C\in R \}.
\end{align*}
Then, define an upper bound $m_A$ of $\Grd{G,A}(w)$ as
\begin{equation*}
  m_A =
      \max_{a\in\Sigma}\bigl(|N_{A,a}| + \Sign{|T_{A,a}|}\bigr)
      + |U_A|,
\end{equation*}
where $\Sign{x}=1$ if $x>0$ and $\Sign{x}=0$ otherwise,
because $\Grd{G,A}(w)$ is never greater than
the number of legal moves from a position $\{(A,w)\}$.

%%%

\subsection{Construction of $M_{A,i}$}
Let a given RLG be $G=(\Gamma,\Sigma,R,I)$ and
let the constructed DFA be $M_{A,i}=(Q,\Sigma,\delta,q_0,F_{A,i})$
for each $A\in\Gamma$ and $i\le m_A$.
Note that the Grundy number of a position
is determined by the ones of the successor positions of it,
and in the SDG on an RLG,
the word $w'$ in
a successor position $\{(B,w')\}$ of position $\{(A,w)\}$
should be a suffix of~$w$.
Hence, we design $M_{A,i}$ as a DFA that
reads a word $w$ from right to left and
keeps track of the Grundy number of
position $\{(B,w')\}$ for each nonterminal~$B$ and
each suffix $w'$ of~$w$.

The construction of
%$Q$, $q_0$, $\delta$, and $F_{A,i}$
$M_{A,i}$
is performed through the following five steps.

\begin{enumerate}
\item Divide the rule set $R$ into the following three sets
  $\Rn$, $\Ru$, and $\Rt$.
  Moreover, introduce a new symbol~$\E$ and define
  a set $\Rt'$ of modified terminal rules as follows.
  We consider $\E$ as a nonterminal and let $m_{\E}=0$.
\begin{align*}
  \Ru &= \{ A\to B  \in R \mid B\in\Gamma \}, \text{ the set of unit rules}, \\
  \Rt &= \{ A\to x  \in R \mid x\in\Sigma^+ \}, \text{ the set of terminal rules}, \\
  \Rn &= \{ A\to xB \in R \mid x\in\Sigma^+,\ B\in\Gamma \}
       = R\setminus(\Ru\cup\Rt).
  \\[\medskipamount]
  \Rt' &= \{ A\to x\E \mid A\to x\in \Rt \}.
\end{align*}

  %According to the set $\Ru$ of unit rules,
  Define a total order $<_{\Gamma}$ over $\Gamma$
  that satisfies $B <_{\Gamma} A$ whenever $A\to B\in\Ru$.
  (Such ordering can be obtained by
  topological sorting of the directed graph
  $(\Gamma,\Ru)$.
  Since unit rules of $G$ do not form a cycle,
  $(\Gamma,\Ru)$ is acyclic.)

\smallskip

\item Define the following $\Gamma'$ and $\GammaN$,
  and
  let $Q=\powset{\Gamma'\cup\GammaN}$
  and $q_0=\{\E^0\}$.
\begin{align*}
  \Gamma' &= \{ A^i \mid A\in\Gamma\cup\{\E\},\ 0\le i\le m_A \}, \\
  \GammaN &= \{ xB^i \mid A\to yxB \in \Rn\cup\Rt',\
                y,x\in\Sigma^+,\ B^i\in\Gamma' \}.
\end{align*}
%
%\begin{align*}
%  Q &= \powset{\Gamma'\cup\GammaN}, \qquad
%  q_0 = \{\E^0\}.
%\end{align*}

\item
  For $q\in Q$, $a\in\Sigma$, and $A\in\Gamma$,
  define the following $J_{A,q,a}$ and $\delta'_A(q,a)$.
  We apply this step to $A\in\Gamma$ in the order
  according to~$<_{\Gamma}$;
  thus, the value of $\delta'_B(q,a)$
  for every $B$ satisfying $A\to B\in\Ru$
  has already been determined.
\begin{align*}
  J_{A,q,a} ={} & \{ j \mid A\to axB \in \Rn\cup\Rt',\ x\in\Sigma^*,\
                     xB^j\in q \} \\
       {}\cup{} & \{ j \mid A\to B\in \Ru,\ B^j \in\delta'_B(q,a) \}. \\[\medskipamount]
  \delta'_A(q,a) ={} &
    \begin{cases}
      \emptyset                 & \text{if } J_{A,q,a} = \emptyset \\
      \{ A^{\mex(J_{A,q,a})} \} & \text{otherwise}. \\
    \end{cases}
\end{align*}

\item
  For $q\in Q$ and $a\in\Sigma$, define $\delta(q,a)$ as follows.
\begin{align*}
  \delta(q,a) &= \bigcup_{A\in\Gamma} \delta'_A(q,a) \cup
        \{ axB^j \in \GammaN \mid x\in\Sigma^*,\ xB^j\in q \}.
\end{align*}

\item
  Let $F_{A,i}=\{q\in Q\mid A^i\in q\}$
  for $A\in\Gamma$ and $i\le m_A$.
\end{enumerate}

%%%

%\hrule

\bigskip

\begin{figure}[t] \centering
  \begin{tikzpicture}[-, scale=1.0, transform shape]
    \tikzstyle{every node}=[font=\small]
    \draw (-2, 0) rectangle (3, 0.4);
    \draw (0.4, 0) -- +(0, 0.4);
    \draw (0, 0) -- +(0, 0.4);
    \node at (-1, 0.2) {$\cdots$};
    \node at (0.2, 0.2) {$a$};
    \node at (1.7, 0.12) {$\smash{w'}$};
    \draw (0, 0.5) -- ++(0, 0.1) -- ++(3, 0) -- ++(0, -0.1);
    \node[anchor=south] at (1.5, 0.6) {$w$};
    \draw[->] (3, 1) -- +(-1, 0);
    \node[anchor=south, scale=0.8, transform shape] at (2.5, 1.05)
      {read from right to left};
    \draw[->] (0, -0.4) -- +(0, 0.3);
    \draw[->] (0.4, -0.4) -- +(0, 0.3);
    \draw[->] (3, -0.4) -- +(0, 0.3);
    \node[anchor=north] (deltaqa) at (-0.3, -0.37) {$\delta(q,a)$};
    \node[anchor=north] at (0.4, -0.4) {$q$};
    \node[anchor=north] at (3, -0.4) {$q_0$};
    \node[anchor=north] at (-0.28, -1.17) {$\delta'_A(q,a)$};
    \node[anchor=north west] (Amex) at (-1, -1.97) {$A^{\mex(J_{\!A,q,a})}$};
    \node[rotate=90] (subseteq) at (-0.3, -1.07) {$\subseteq$};
    \node[rotate=90] (In) at (-0.3, -1.87) {$\in$};
    \node[font=\tiny, right=1.7 of Amex.center] (b1) {$\bullet$};
    \node[font=\tiny, above=0.2 of b1] (b2) {$\bullet$};
    \node[font=\tiny, below=0.2 of b1] (b3) {$\bullet$};
    \foreach \b in {b1, b2, b3} {
      \draw[->] (Amex) decorate[SharpSquiggly]{ -- ($(Amex)!0.8!(\b)$) } -- (\b);
    }
    \draw[decorate,decoration={brace,amplitude=5pt}]
      (2.1, -1.6) -- +(0, -1.25);
    \node[right=0.3 of b1] (JAqa) {$J_{A,q,a}$};
    \node[scale=0.72, transform shape] at (0.2, -2.9) {$\Grd{G}(\{(A,w)\})$};
    \node[rotate=90] at (0, -2.55) {$=$};
    \node[scale=0.72, transform shape, below=0.4 of JAqa.west, anchor=west]
      {${}=\{\Grd{G}(P)\mid \{(A,w)\}\posTo P\}$};
  \end{tikzpicture}
  \caption{The state of $M_{X,i}$ just after reading $\Rev{(aw')}$.}
  \label{fig:M_{X,i}}
\end{figure}
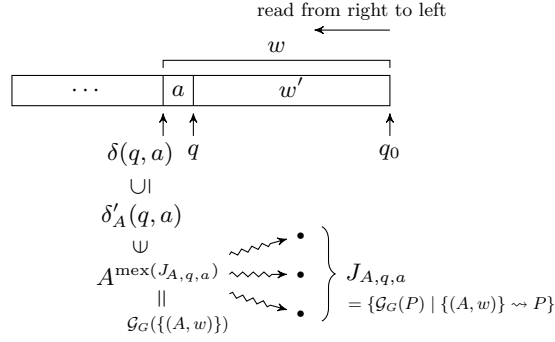%

The above DFA $M_{X,i}$
is designed based on the following idea:
Assume that $M_{X,i}$ has just read a suffix $w'$ of
some word from right to left
(or equivalently read $\Rev{(w')}$ from left to right)
and reached a state~$q$,
and then it is reading a symbol~$a$
(Fig.~\ref{fig:M_{X,i}}).
Let $w=aw'$.
Then, $J_{A,q,a}$ in Step~3 represents
the set of the Grundy numbers of the successor positions
of position~$\{(A,w)\}$.
$M_{X,i}$ has already read every suffix of~$w'$ and
keeps information on the Grundy numbers of the successor positions
of $\{(A,w)\}$ inside the state~$q$.
$J_{A,q,a}$ can be computed using the information in~$q$ and
the rules in $\Rn\cup\Rt'\cup\Ru$.
Set $\delta'_A(q,a)$,
which is either a singleton set or the empty set,
represents
the Grundy number of $\{(A,w)\}$ determined by $J_{A,q,a}$.
($J_{A,q,a}$ and $\delta'_A(q,a)$ become empty
iff $\{(A,w)\}$ is not a position,
i.e., $w\notin L(G,A)$.)
In Step~4, the information in
$\delta'_A(q,a)$ for every $A\in\Gamma$ is gathered
into $\delta(q,a)$,
with some elements of $\GammaN$ each of which represents
an intermediate state of a single production rule
and is possibly used later for computing~$J$.

%%% ----- Correctness -----

\section{Correctness of the Construction of DFAs}
\label{sec:correctness}

We fix a given RLG
$G=(\Gamma,\Sigma,R,I)$ in this section.
Let $M_{A,i} =
(Q,\Sigma,\allowbreak \delta,\allowbreak q_0,\allowbreak F_{A,i})$
for $A\in\Gamma$ and $i\le m_A$ be the DFA
constructed from~$G$
by the algorithm in Section~\ref{sec:RLGtoDFA}.
Let $\delta^* : Q\times\Sigma^*\to Q$ be the extension
of $\delta$ to strings; i.e.,
$\delta^*(q,\varepsilon) = q$ and
$\delta^*(q,xa) = \delta(\delta^*(q,x),a)$
for $q\in Q$, $x\in\Sigma^*$, and $a\in\Sigma$.

For convenience of proof,
we consider $\{(\E,\varepsilon)\}$ is a position and
is equivalent to the final empty position;
i.e.,
$\{(A,w)\} \posTo \{(\E,\varepsilon)\}$
if and only if $A\to w\in\Rt$.
Moreover, $\Grd{G,\E}(\varepsilon)=0$ and
$\Grd{G,\E}(w)=\bot$ for any $w\in\Sigma^+$.

\begin{lemma}\label{lemma:GammaN}
  For $w\in\Sigma^*$, $B\in\Gamma\cup\{\E\}$,
  $j\le m_B$,
  $x\in\Sigma^+$ such that $xB^j\in\GammaN$,
  \begin{equation*}
    xB^j\in\delta^*(q_0,w) \iff
    B^j\in\delta^*(q_0,w\RQ(\Rev{x})).
  \end{equation*}
\end{lemma}
\begin{proof}
  Let $k=|x|$ and $x=x_1\ldots x_k$ where
  $x_i\in\Sigma$ for $1\le i\le k$.
  Assume that $xB^j\in\GammaN$.
  By the definition of $\GammaN$,
  $xB^j\in\GammaN$ implies
  $x_2\ldots x_k B^j\in\GammaN$, \ldots,
  $x_k B^j\in \GammaN$.

  \smallskip
  (Only-if part) \
  Assume that $xB^j\in\delta^*(q_0,w)$.
  Since $q_0=\{\E^0\}$ and $x\in\Sigma^+$,
  $w\ne\varepsilon$ holds.
  Let $w=w'a$ where
  $w'\in\Sigma^*$ and $a\in\Sigma$.
  Since $\delta^*(q_0,w)=\delta(\delta^*(q_0,w'),a)$ and
  by the definition of $\delta$,
  the following equivalence holds.
  \begin{equation*}
    xB^j\in\delta^*(q_0,w) \iff
    x_1 = a \:\land\:
    x_2\ldots x_k B^j \in \delta^*(q_0,w').
  \end{equation*}
  Applying the same expansion to
  $x_2\ldots x_k B^j \in \delta^*(q_0,w')$
  repeatedly,
  we obtain the following equivalence.
  \begin{align*}
    xB^j\in\delta^*(q_0,w) \iff{} &
    \exists w''.\ w = w''\Rev{x} \:\land\:
    B^j \in \delta^*(q_0,w'').
  \end{align*}
  The right-hand side implies
  $B^j\in\delta^*(q_0,w\RQ(\Rev{x}))$.

  \smallskip
  (If-part) \
  Assume that
  $B^j\in\delta^*(q_0,w\RQ(\Rev{x}))$.
  This means that $w=w''\Rev{x}$ for some $w''\in\Sigma^*$.
  Therefore,
  by traversing the above equivalence in the only-if part
  conversely,
  we obtain $xB^j\in\delta^*(q_0,w)$.
  \qed
\end{proof}

\begin{lemma}\label{lemma:main}
  For $w\in\Sigma^*$, $A\in\Gamma\cup\{\E\}$, $i\le m_A$,
  \begin{equation*}
    A^i\in \delta^*(q_0,w) \iff \Grd{G,A}(\Rev{w})=i.
  \end{equation*}
\end{lemma}
\begin{proof}
  We prove this lemma by induction on the length of~$w$.

  \smallskip
  (Base case) \
  Since $\delta^*(q_0,\varepsilon)=q_0=\{\E^0\}$,
  $A^i\in\delta^*(q_0,\varepsilon)$ implies
  $A=\E$ and $i=0$,
  and hence $\Grd{G,A}(\Rev{w})=\Grd{G,\E}(\varepsilon)=i$
  holds.
  For the converse,
  $\Grd{G,A}(\varepsilon)\ne\bot$ implies
  $A=\E$ and $\Grd{G,A}(\varepsilon)=0$,
  and hence $A^i=\E^0\in\delta^*(q_0,\varepsilon)$ holds.

  \smallskip
  (Inductive step) \
  Assume $|w|>0$.
  As the induction hypothesis,
  we assume that for every $y\in\Sigma^*$ with $|y|<|w|$,
  $B\in\Gamma\cup\{\E\}$, and $j\le m_B$,
  \begin{equation}
    B^j\in \delta^*(q_0,y) \iff \Grd{G,B}(\Rev{y})=j.
    \label{eq:IH1}
  \end{equation}

  When $|w|>0$,
  $A^i\in\delta^*(q_0,w)$ implies $A\ne\E$ and
  so does $\Grd{G,A}(\Rev{w})\ne\bot$.
  Hence, we assume $A\ne\E$ in the following.

  Let $w=w'a$ for $w'\in\Sigma^*$ and $a\in\Sigma$.
  Also let $q=\delta^*(q_0,w')$.
  Note that $\delta^*(q_0,w)=\delta(q,a)$.
  We define $S$ as the set of the successor positions
  of $\{(A,\Rev{w})\}$; i.e.,
  $S=\{s\mid \{(A,\Rev{w})\} \posTo s \}$.
  We show
  $A^i\in \delta^*(q_0,w) \iff \Grd{G,A}(\Rev{w})=i$
  by induction on the order of $A$ according to~$<_{\Gamma}$.
  A key point is to show that
  $J_{A,q,a}=\{\Grd{G}(s)\mid s\in S\}$ and thus
  $\Grd{G}(\{(A,\Rev{w})\})=\mex(J_{A,q,a})$.

  \smallskip
  (Base case of the inner induction) \
  Assume that there is no $B\in\Gamma$ such that $A\to B\in\Ru$.
  Hence, $J_{A,q,a}=\{j\mid A\to axB\in\Rn\cup\Rt'$,
  $x\in\Sigma^*$, $xB^j\in q\}$.

  Assume that $\{(B,y)\}\in S$ for $B\in\Gamma\cup\{\E\}$
  and $y\in\Sigma^*$.
  By the definition of the game,
  $\{(B,y)\}\in S$ if and only if
  $A\to axB\in \Rn\cup\Rt'$ for some $x\in\Sigma^*$ such that
  $\Rev{w}=axy$ and $\Grd{G,B}(y)\ne\bot$.
  Let $j=\Grd{G,B}(y)$.
  By the induction hypothesis~(\ref{eq:IH1}),
  $\Grd{G,B}(y)\ne\bot$ implies that
  $B^j\in \delta^*(q_0,\Rev{y})$.
  If $x\ne\varepsilon$, then
  $A\to axB\in\Rn\cup\Rt'$ implies $xB^j\in\GammaN$,
  %by the definition of~$\GammaN$
  which is followed by
  $xB^j\in\delta^*(q_0,w')$
  by Lemma~\ref{lemma:GammaN}.
  If $x=\varepsilon$, then
  $w'=\Rev{y}$ and thus $xB^j\in\delta^*(q_0,w')$ holds.
  Therefore,
  $\Grd{G,B}(y)\in J_{A,q,a}$ holds in any case,
  and hence $\{\Grd{G}(s)\mid s\in S\}\subseteq J_{A,q,a}$.

  For the converse, assume that $j\in J_{A,q,a}$.
  Then,
  $A\to axB\in\Rn\cup\Rt'$ for some $B\in\Gamma\cup\{\E\}$ and
  $x\in\Sigma^*$ and $xB^j\in q=\delta^*(q_0,w')$.
  If $x=\varepsilon$,
  then $\Grd{G,B}(\Rev{(w')})=j$ holds
  by the induction hypothesis~(\ref{eq:IH1}).
  If $x\ne\varepsilon$,
  then $xB^j\in\GammaN$, which is followed by
  $B^j\in\delta^*(q_0,w'\RQ(\Rev{x}))$
  by Lemma~\ref{lemma:GammaN}.
  This implies $\Grd{G,B}(x\LQ\Rev{(w')})=j$
  by~(\ref{eq:IH1}).
  In any case,
  $\Grd{G,B}(y)=j\ne\bot$
  for $y=(ax)\LQ(\Rev{w})$,
  and thus $\{(B,y)\}\in S$.
  Therefore, $J_{A,q,a}\subseteq \{\Grd{G}(s)\mid s\in S\}$.
  Since we have shown
  $J_{A,q,a}=\{\Grd{G}(s)\mid s\in S\}$,
  $\Grd{G}(\{(A,\Rev{w})\})=\mex(J_{A,q,a})$ holds.

  Assume that $A^i\in\delta^*(q_0,w)=\delta(q,a)$.
  By the definition of $\delta$,
  $A^i\in\delta'_A(q,a)$ holds.
  By the definition of $\delta'_A$,
  $J_{A,q,a}\ne\emptyset$ and $\mex(J_{A,q,a})=i$,
  which is followed by
  $\Grd{G,A}(\Rev{w})=i$.
  For the converse,
  assume that $\Grd{G,A}(\Rev{w})=i$.
  Since $\Grd{G,A}(\Rev{w})\ne\bot$,
  $\Rev{w}\in L(G,A)$ and thus $S\ne\emptyset$.
  Hence, $J_{A,q,a}\ne\emptyset$ and
  $\delta'_A(q,a)=\{A^i\}$,
  which is followed by $A^i\in\delta(q,a)$.

  \smallskip
  (Inductive step of the inner induction) \
  As the induction hypothesis,
  we assume that for every
  $B\in\Gamma$ with $B <_{\Gamma} A$ and $j\le m_B$,
  \begin{equation}
    B^j\in \delta^*(q_0,w) \iff \Grd{G,B}(\Rev{w})=j.
    \label{eq:IH2}
  \end{equation}
  Note that
  by the definition of $\delta$,
  $B^j\in\delta^*(q_0,w)$ if and only if
  $B^j\in\delta'_B(q,a)$.

  By definition,
  $J_{A,q,a}=\{j\mid A\to axB\in\Rn\cup\Rt'$,
  $x\in\Sigma^*$, $xB^j\in q\}\cup\{j\mid A\to B\in\Ru$,
  $B^j\in\delta'_B(q,a)\}$.
  In a similar way to the base case,
  we can show that $J_{A,q,a}=\{\Grd{G}(s)\mid s\in S\}$,
  and thus
  $\Grd{G}(\{(A,\Rev{w})\})=\mex(J_{A,q,a})$ holds.
  In the same way as the base case,
  we obtain $A^i\in\delta(q,a) \iff \Grd{G,A}(\Rev{w})=i$.
  \qed
\end{proof}

\begin{theorem}
  \label{theorem:main}
  For $w\in\Sigma^+$, $A\in\Gamma$, $i\le m_A$,
  \begin{equation*}
    %A^i\in \delta^*(q_0,\Rev{w})
    \Rev{w}\in L(M_{A,i}) \iff \Grd{G,A}(w)=i.
  \end{equation*}
\end{theorem}
\begin{proof}
  This follows directly from Lemma~\ref{lemma:main},
  since $\Rev{w}\in L(M_{A,i})$ if and only if
  $A^i\in \delta^*(q_0,\Rev{w})$
  by the construction of $F_{A,i}$.
  \qed
\end{proof}

\begin{corollary}
  For any $A\in\Gamma$ and $c\le m_A$,
  the set $\{w\in\Sigma^+ \mid \Grd{G,A}(w)=c\}$,
  which represents the set of positions whose
  nonterminal parts are $A$ and Grundy numbers equal~$c$,
  is regular.
\end{corollary}
\begin{proof}
  This follows directly from Theorem~\ref{theorem:main},
  since the class of regular languages is closed under the reverse operation.
  \qed
\end{proof}

\begin{corollary}
  \label{coro:decidability-of-maxG}
  \textsc{\textsl{Max-Grundy-Number-in-SDG}}
  where the input CFG $G$ is restricted to be right-linear,
  described below,
  is decidable.
  \begin{itemize}
  \item Instance: an RLG $G=(\Gamma,\Sigma,R,I)$,
    a nonterminal $A\in\Gamma$, and $k\in\Nat$.
  \item Question: $\max_{w\in L(G,A)}\bigl(\Grd{G,A}(w)\bigr) = k$ ?
  \end{itemize}
\end{corollary}
\begin{proof}
  We can answer this problem by checking whether
  $L(M_{A,k})\ne\emptyset$ and
  $L(M_{A,i})=\emptyset$
  for each $i$ with $k < i\le m_A$.
  \qed
\end{proof}

%%% ----- PSPACE-completeness -----

\section{PSPACE-Completeness of
  %Computing the Least Upper Bound of Grundy Numbers in the SDG
  \textsc{\textbf{Max-Grundy-Number-in-SDG}}
  on RLGs}
\label{sec:complexity}

\begin{lemma}
  \label{lemma:in-PSPACE}
  \textsc{\textsl{Max-Grundy-Number-in-SDG}}
  is in PSPACE
  if the input CFG is restricted to be right-linear.
\end{lemma}
\begin{proof}
  By the proof of Corollary~\ref{coro:decidability-of-maxG},
  it is sufficient to show
  the PSPACE-solvability of deciding $L(M_{A,i})\ne\emptyset$
  for each $k\le i\le m_A$.
  Since $\mathrm{PSPACE}=\mathrm{NPSPACE}$,
  we show a non-deterministic polynomial-space algorithm
  for deciding $L(M_{A,i})\ne\emptyset$.

  Let $G=(\Gamma,\Sigma,R,I)$ be a given RLG and
  $M_{A,i}=(Q,\Sigma,\delta,q_0,F_{A,i})$ be the constructed
  DFA for $G$, $A\in\Gamma$ and $i\le m_A$.
  Note that $m_A \le |R|$ for each $A\in\Gamma$
  by the discussion in Section~\ref{sec:m_A}.
  Obviously, $L(M_{A,i})\ne\emptyset$ if and only if
  there exists a path $q_0q_1\ldots q_n$ of states
  following the transition function $\delta$,
  starting from $q_0$ and ending at $q_n\in F_{A,i}$,
  such that $n < |Q|$.
  Since $Q=\powset{\Gamma'\cup\GammaN}$ and
  $|\Gamma'|+|\GammaN|$ is a polynomial of
  the description length of $G$,
  each $q\in Q$ can be stored in polynomial space.
  The algorithm first writes $q_0$
  into the tape and
  repeatedly updates state $q$ on the tape to %the next state
  $\delta(q,a)$
  by non-deterministically choosing $a\in\Sigma$.
  %The algorithm repeats the update at most $|Q|$ times.
  If it reaches a state in $F_{A,i}$
  within $|Q|$ update steps,
  then accept; otherwise, reject.
  The counter for counting the number of repetition
  at most $|Q|$ in binary also requires only polynomial space.
  Therefore, the algorithm is
  a non-deterministic polynomial-space one
  that decides whether $L(M_{A,i})\ne\emptyset$.
  \qed
\end{proof}

\begin{theorem}
  \label{lemma:PSPACE-comp}
  \textsc{\textsl{Max-Grundy-Number-in-SDG}}
  where the input CFG is restricted to be right-linear
  is PSPACE-complete.
\end{theorem}
\begin{proof}
  By Lemma~\ref{lemma:in-PSPACE},
  it is sufficient to show the PSPACE-hardness of
  the problem.
  We show a polynomial-time reduction from
  the following PSPACE-complete problem~\cite{AHU74,GJ79}.
  %\cite[Sect.~10.6]{AHU74}
  %
  \begin{itemize}
  \item[] \hspace{-1em}%
    \textsc{Regular-Grammar-Non-Universality}
  \item Instance: a regular grammar $G=(\Gamma,\Sigma,R,I)$
    (without $\varepsilon$-rules).
  \item Question: $L(G) \ne \Sigma^+$ ?
  \end{itemize}
  A \emph{regular grammar} (RG) is a right-linear grammar where
  the right-hand side of each rule is an element of
  $\Sigma\cup\Sigma\Gamma\cup\Gamma$.
  Note that the definition of a CFG in this paper does not allow
  $\varepsilon$-rules, i.e., rules of the form $A\to\varepsilon$.
  Because checking whether $\varepsilon\in L(G)$ and
  eliminating $\varepsilon$-rules
  in an RG with $\varepsilon$-rules
  can be performed in polynomial time,
  the problem is still PSPACE-complete for
  RGs without $\varepsilon$-rules.

  Let $G=(\Gamma,\Sigma,R,I)$ be an input RG for
  \textsc{Regular-Grammar-Non-Univer\-sality}.
  Let $G_{\ALL}=(\{I_{\ALL}\},\Sigma,R_{\ALL},I_{\ALL})$
  be an RG
  %$\Gamma_{\ALL}=\{I_{\ALL}\}$ and
  %where $R_{\ALL}=
  %  \{ I_{\ALL} \to a I_{\ALL} \mid a\in\Sigma \}\cup
  %  \{ I_{\ALL} \to a \mid a\in\Sigma \}$;
  such that $L(G_{\ALL})=\Sigma^+$.
  The reduction to \textsc{Max-Grundy-Number-in-SDG}
  on RLGs constructs an RLG
  $G'=(\Gamma\cup\{I_{\ALL},I_1,\ldots,I_9\},
  \Sigma, R\cup R_{\ALL}\cup R', I_1)$
  where
  $R'$ consists of twelve unit rules shown in
  the following diagram.
  %$R'=\{I_1\to I_2$, $I_2\to I_3\mid I_7$,
  %$I_3\to I_4$, $I_4\to I_5$, $I_5\to I_6\mid I_{\ALL}$,
  %$I_6\to I_{\ALL}$,
  %$I_7\to I_8$, $I_8\to I_9\mid I$, $I_9\to I\}$.
  %
  
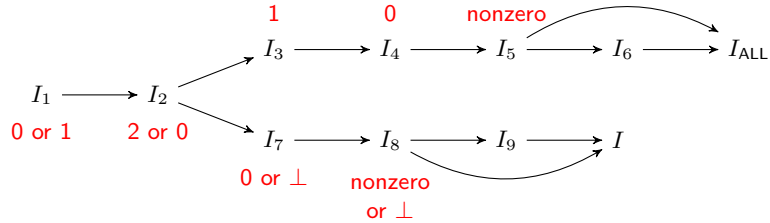
\begin{center}
  \begin{tikzpicture}
    \node (I1) {$I_1$};
    \node[right=of I1] (I2) {$I_2$};
    \node[above right=0.1 and 1.0 of I2] (I3) {$I_3$};
    \node[right=of I3] (I4) {$I_4$};
    \node[right=of I4] (I5) {$I_5$};
    \node[right=of I5] (I6) {$I_6$};
    \node[right=of I6] (Iall) {$I_{\ALL}$};
    \node[below right=0.1 and 1.0 of I2] (I7) {$I_7$};
    \node[right=of I7] (I8) {$I_8$};
    \node[right=of I8] (I9) {$I_9$};
    \node[right=of I9] (I) {$I$};
    \graph {
      (I1) -> (I2) -> (I3) -> (I4) -> (I5) -> (I6) -> (Iall);
      (I2) -> (I7) -> (I8) -> (I9) -> (I);
    };
    \draw[->] (I5) edge[bend left] (Iall)
      (I8) edge[bend right] (I)
    ;
    \tikzstyle{every node}=[font={\footnotesize\sffamily}, text=red]
    \node[above=0 of I5] {nonzero};
    \node[above=0 of I4] {0};
    \node[above=0 of I3] {1};
    \node[below=0 of I8] {\begin{tabular}{c}nonzero \\ or $\bot$\end{tabular}};
    \node[below=0 of I7] {0 or $\bot$};
    \node[below=0 of I2] {2 or 0};
    \node[below=0 of I1] {0 or 1};
  \end{tikzpicture}
  \end{center}
  The diagram also shows the possible Grundy numbers for
  some nonterminals.
  For example, $I_5$ is labeled with ``nonzero,''
  which means that $\Grd{G',I_5}(w) > 0$
  for every $w\in\Sigma^+$.
  This is implied by the fact that
  either $\Grd{G',I_6}(w)=0$ or $\Grd{G',I_{\ALL}}(w)=0$ holds
  for every $w\in\Sigma^+$,
  because $\Grd{G',I_6}(w)=0$ whenever
  $\Grd{G',I_{\ALL}}(w) > 0$.
  (This construction is based on Lemma~2 of \cite{IKIK26}.)
  The fact that
  $\Grd{G',I_5}(w) > 0$ for every $w\in\Sigma^+$
  implies $\Grd{G',I_4}(w)=0$ and $\Grd{G',I_3}(w)=1$
  for every $w\in\Sigma^+$.
  Similarly,
  $\Grd{G',I_8}(w)>0$ and $\Grd{G',I_7}(w)=0$ for
  every $w\in L(G',I) = L(G)$,
  and $\Grd{G',I_7}(w)=\Grd{G',I_8}(w)=\bot$ for
  every $w\notin L(G)$.
  Since $\{(I_2,w)\}$ has two successors
  $\{(I_3,w)\}$ and $\{(I_7,w)\}$ if $w\in L(G)$ and
  has only one successor $\{(I_3,w)\}$ if $w\notin L(G)$,
  %$\Grd{G',I_2}(w) = 2$ and $\Grd{G',I_1}(w) = 0$
  %if $w\in L(G)$ and
  %$\Grd{G',I_2}(w) = 0$ and $\Grd{G',I_1}(w) = 1$
  %if $w\notin L(G)$.
  \begin{equation*}
    \Grd{G',I_2}(w) = \begin{cases}
      2 & \text{if } w\in L(G) \\
      0 & \text{otherwise},
    \end{cases}
    \ \ \text{and}\ \
    \Grd{G',I_1}(w) = \begin{cases}
      0 & \text{if } w\in L(G) \\
      1 & \text{otherwise}.
    \end{cases}
  \end{equation*}
  Therefore, $\max_{w\in\Sigma^+}\bigl(\Grd{G',I_1}(w)\bigr)=1$
  if and only if there exists
  $w\notin L(G)$.
  %$w\in \Sigma^+\setminus L(G)$.
  \qed
\end{proof}

%%% ----- Examples -----

\section{Examples}
\label{sec:examples}

\begin{example}
\label{example:csub}
  Consider the same RLG as Example~\ref{example:pre-csub}:
  $G=(\{A\},\allowbreak \Sigma,\allowbreak R,A)$ where $\Sigma=\{\t{a},\t{b}\}$ and
  $R=\{\,
  A \to \t{a}A \mid \t{aa}A \mid \t{b}A \mid \t{a}
  \mid \t{aa} \mid \t{b} \,\}$.
  %which equals the grammar $\mathbf{CSub}$
  %described in Examples~4 and 7 in~\cite{IKIK26}.
%
  Let $m_A=3$.
  The construction of $M_{A,i}=(Q,\Sigma,\delta,q_0,F_{A,i})$
  for $i\le m_A$ is
  performed as follows.

  \medskip
  \noindent
  Step~1:\quad
  Let $\Rn=\{\, A\to \t{a}A \mid \t{aa}A \mid \t{b}A \,\}$,
  $\Rt'=\{\, A\to \t{a}\E \mid \t{aa}\E \mid \t{b}\E \,\}$,
  and $\Ru=\emptyset$.

  \smallskip
  \noindent
  Step~2:\quad
  Let $\Gamma'=\{A^0,\ldots,A^3,\E^0\}$ and
  $\GammaN=\{\t{a}A^0,\ldots,\t{a}A^3,\t{a}\E^0\}$.
  %Then, let $Q=\powset{\Gamma'\cup\GammaN}$ and
  Moreover, let $q_0=\{\E^0\}$.

  \smallskip
  \noindent
  Steps~3 and~4:\quad
  For simplicity,
  we restrict the state set $Q$ to the set of states
  reachable from~$q_0$, and
  we define the transition function only on the restricted~$Q$.
  Starting from $q_0$,
  we iteratively apply Steps~3 and~4 to a newly found state~$q$.
  We obtain the following
  $J_{A,q,a}$, $\delta'_A(q,a)$, and $\delta(q,a)$
  for state $q$ reachable from $q_0$ and $a\in\{\t{a},\t{b}\}$.
\bgroup\small
\begin{align*}
  J_{A,         \{\E^0\},\t{a}} &= \{0\},
    & \delta'_A(\{\E^0\},\t{a}) &= \{A^1\},
    & \delta(   \{\E^0\},\t{a}) &= \{A^1,\t{a}\E^0\}. \\
  J_{A,         \{\E^0\},\t{b}} &= \{0\},
    & \delta'_A(\{\E^0\},\t{b}) &= \{A^1\},
    & \delta(   \{\E^0\},\t{b}) &= \{A^1\}. \\
  J_{A,         \{A^1,\t{a}\E^0\},\t{a}} &= \{1,0\},
    & \delta'_A(\{A^1,\t{a}\E^0\},\t{a}) &= \{A^2\},
    & \delta(   \{A^1,\t{a}\E^0\},\t{a}) &= \{A^2,\t{a}A^1\}. \\
  J_{A,         \{A^1,\t{a}\E^0\},\t{b}} &= \{1\},
    & \delta'_A(\{A^1,\t{a}\E^0\},\t{b}) &= \{A^0\},
    & \delta(   \{A^1,\t{a}\E^0\},\t{b}) &= \{A^0\}. \\
  J_{A,         \{A^1\},\t{a}} &= \{1\},
    & \delta'_A(\{A^1\},\t{a}) &= \{A^0\},
    & \delta(   \{A^1\},\t{a}) &= \{A^0,\t{a}A^1\}. \\
  J_{A,         \{A^1\},\t{b}} &= \{1\},
    & \delta'_A(\{A^1\},\t{b}) &= \{A^0\},
    & \delta(   \{A^1\},\t{b}) &= \{A^0\}. \\
  J_{A,         \{A^2,\t{a}A^1\},\t{a}} &= \{2,1\},
    & \delta'_A(\{A^2,\t{a}A^1\},\t{a}) &= \{A^0\},
    & \delta(   \{A^2,\t{a}A^1\},\t{a}) &= \{A^0,\t{a}A^2\}. \\
  J_{A,         \{A^2,\t{a}A^1\},\t{b}} &= \{2\},
    & \delta'_A(\{A^2,\t{a}A^1\},\t{b}) &= \{A^0\},
    & \delta(   \{A^2,\t{a}A^1\},\t{b}) &= \{A^0\}. \\
  J_{A,         \{A^0\},\t{a}} &= \{0\},
    & \delta'_A(\{A^0\},\t{a}) &= \{A^1\},
    & \delta(   \{A^0\},\t{a}) &= \{A^1,\t{a}A^0\}. \\
  J_{A,         \{A^0\},\t{b}} &= \{0\},
    & \delta'_A(\{A^0\},\t{b}) &= \{A^1\},
    & \delta(   \{A^0\},\t{b}) &= \{A^1\}. \\
  J_{A,         \{A^0,\t{a}A^1\},\t{a}} &= \{0,1\},
    & \delta'_A(\{A^0,\t{a}A^1\},\t{a}) &= \{A^2\},
    & \delta(   \{A^0,\t{a}A^1\},\t{a}) &= \{A^2,\t{a}A^0\}. \\
  J_{A,         \{A^0,\t{a}A^1\},\t{b}} &= \{0\},
    & \delta'_A(\{A^0,\t{a}A^1\},\t{b}) &= \{A^1\},
    & \delta(   \{A^0,\t{a}A^1\},\t{b}) &= \{A^1\}. \\
  J_{A,         \{A^0,\t{a}A^2\},\t{a}} &= \{0,2\},
    & \delta'_A(\{A^0,\t{a}A^2\},\t{a}) &= \{A^1\},
    & \delta(   \{A^0,\t{a}A^2\},\t{a}) &= \{A^1,\t{a}A^0\}. \\
  J_{A,         \{A^0,\t{a}A^2\},\t{b}} &= \{0\},
    & \delta'_A(\{A^0,\t{a}A^2\},\t{b}) &= \{A^1\},
    & \delta(   \{A^0,\t{a}A^2\},\t{b}) &= \{A^1\}. \\
  J_{A,         \{A^1,\t{a}A^0\},\t{a}} &= \{1,0\},
    & \delta'_A(\{A^1,\t{a}A^0\},\t{a}) &= \{A^2\},
    & \delta(   \{A^1,\t{a}A^0\},\t{a}) &= \{A^2,\t{a}A^1\}. \\
  J_{A,         \{A^1,\t{a}A^0\},\t{b}} &= \{1\},
    & \delta'_A(\{A^1,\t{a}A^0\},\t{b}) &= \{A^0\},
    & \delta(   \{A^1,\t{a}A^0\},\t{b}) &= \{A^0\}. \\
  J_{A,         \{A^2,\t{a}A^0\},\t{a}} &= \{2,0\},
    & \delta'_A(\{A^2,\t{a}A^0\},\t{a}) &= \{A^1\},
    & \delta(   \{A^2,\t{a}A^0\},\t{a}) &= \{A^1,\t{a}A^2\}. \\
  J_{A,         \{A^2,\t{a}A^0\},\t{b}} &= \{2\},
    & \delta'_A(\{A^2,\t{a}A^0\},\t{b}) &= \{A^0\},
    & \delta(   \{A^2,\t{a}A^0\},\t{b}) &= \{A^0\}. \\
  J_{A,         \{A^1,\t{a}A^2\},\t{a}} &= \{1,2\},
    & \delta'_A(\{A^1,\t{a}A^2\},\t{a}) &= \{A^0\},
    & \delta(   \{A^1,\t{a}A^2\},\t{a}) &= \{A^0,\t{a}A^1\}. \\
  J_{A,         \{A^1,\t{a}A^2\},\t{b}} &= \{1\},
    & \delta'_A(\{A^1,\t{a}A^2\},\t{b}) &= \{A^0\},
    & \delta(   \{A^1,\t{a}A^2\},\t{b}) &= \{A^0\}. \\
\end{align*}
\egroup

  \noindent
  Step~5:\quad
  Define $F_{A,i}$ for $0\le i\le 3$
  as follows.
  \begin{align*}
    F_{A,0} &= \{\{A^0\}, \{A^0,\t{a}A^1\}, \{A^0,\t{a}A^2\}\}, \\
    F_{A,1} &= \{
      \{A^1\}, \{A^1,\t{a}\E^0\}, \{A^1,\t{a}A^0\}, \{A^1,\t{a}A^2\}
      \}, \\
    F_{A,2} &= \{\{A^2,\t{a}A^0\}, \{A^2,\t{a}A^1\}\}, \\
    F_{A,3} &= \emptyset.
  \end{align*}
  %\medskip

  \begin{figure}[t]\centering
  \begin{tikzpicture}[scale=0.9,transform shape]
    \node[initial,elliptic state] (E0) {$\E^0$};
    \node[elliptic state] (A1aE0) [right=of E0] {$A^1,\t{a}\E^0$};
    \node[elliptic state] (A2aA1) [right=of A1aE0] {$A^2,\t{a}A^1$};
    \node[elliptic state] (A0aA2) [right=of A2aA1] {$A^0,\t{a}A^2$};
    \node[elliptic state] (A1) [below=of A1aE0] {$A^1$};
    \node[elliptic state] (A0) [below=of A2aA1] {$A^0$};
    \node[elliptic state] (A1aA0) [below=of A0aA2] {$A^1,\t{a}A^0$};
    \node[elliptic state] (A0aA1) [below=of A1] {$A^0,\t{a}A^1$};
    \node[elliptic state] (A2aA0) [below=of A0] {$A^2,\t{a}A^0$};
    \node[elliptic state] (A1aA2) [below=of A1aA0] {$A^1,\t{a}A^2$};
    \path
    (E0)    edge node {$\t{a}$} (A1aE0)
            edge [bend right] node {$\t{b}$} (A1)
    (A1aE0) edge node {$\t{a}$} (A2aA1)
            edge node [yshift=-3mm,xshift=4mm] {$\t{b}$} (A0)
    (A2aA1) edge node {$\t{a}$} (A0aA2)
            edge node [yshift=-2mm] {$\t{b}$} (A0)
    (A0aA2) edge node {$\t{a}$} (A1aA0)
            edge [out=-150,in=60] node [swap,xshift=24mm,yshift=-2mm] {$\t{b}$} (A1)
    (A1)    edge [bend right=10] node [swap] {$\t{a}$} (A0aA1)
            edge [bend left=5] node {$\t{b}$} (A0)
    (A0)    edge [bend left=5] node {$\t{a}$} (A1aA0)
            edge [bend left=5] node {$\t{b}$} (A1)
    (A1aA0) edge node [swap,yshift=-3mm,xshift=4mm] {$\t{a}$} (A2aA1)
            edge [bend left=5] node {$\t{b}$} (A0)
    (A0aA1) edge node {$\t{a}$} (A2aA0)
            edge [bend right=10] node [swap] {$\t{b}$} (A1)
    (A2aA0) edge node {$\t{a}$} (A1aA2)
            edge node {$\t{b}$} (A0)
    (A1aA2) edge [bend left=25] node {$\t{a}$} (A0aA1)
            edge node {$\t{b}$} (A0)
    ;
  \end{tikzpicture}
  \caption{DFA $(Q,\Sigma,\delta,q_0,\emptyset)$
    in Example~\ref{example:csub}.}
  \label{fig:csub-r}
  \end{figure}
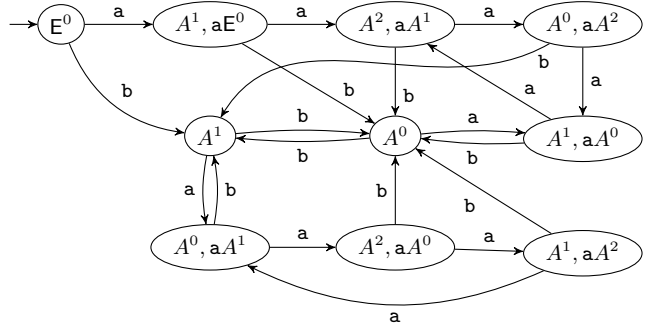

  Figure~\ref{fig:csub-r} shows the constructed DFA
  $(Q,\Sigma,\delta,q_0,\emptyset)$ without final states.
  By simulating the behavior of this DFA
  on the words of length 4 or less,
  we obtain the following table of languages.
  We can confirm that every word $w$ appearing in this table
  satisfies
  ``$w\in L(M_{A,i}) \iff \Grd{G,A}(\Rev{w})=i$.''
  \begin{align*}
    L(M_{A,0})&=\{
      \begin{aligned}[t]
      &\t{ab},\t{ba},\t{bb},\t{aaa},\t{aab},
      \t{abab},\t{abba},\t{abbb},\t{baab},\t{baba},\t{babb},\t{bbab}, \\
      &\t{bbba},\t{bbbb},
      \ldots
    \}, \end{aligned} \\
    L(M_{A,1})&=\{
      \t{a},\t{b},\t{aba},\t{abb},\t{bab},\t{bba},\t{bbb},
      \t{aaaa},\t{aaab},\t{aaba},\t{aabb},\t{baaa},
      \ldots
    \}, \\
    L(M_{A,2})&=\{
      \t{aa},\t{baa},\t{abaa},\t{bbaa},\ldots
    \}, \\
    L(M_{A,3})&=\emptyset.
  \end{align*}

  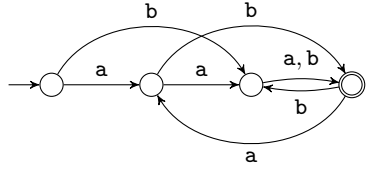
\begin{figure}[t]\centering
  \begin{tikzpicture}
    \node[initial,state] (E0) {};
    \node[state] (q1) [right=of E0] {};
    \node[state] (q2) [right=of q1] {};
    \node[state,accepting] (q0) [right=of q2] {};
    \path
    (E0) edge node {$\t{a}$} (q1)
         edge [bend left=60] node {$\t{b}$} (q2)
    (q1) edge node {$\t{a}$} (q2)
         edge [bend left=60] node {$\t{b}$} (q0)
    (q2) edge [bend left=10] node {$\t{a},\t{b}$} (q0)
    (q0) edge [bend left=60] node {$\t{a}$} (q1)
         edge [bend left=10]node {$\t{b}$} (q2)
    ;
  \end{tikzpicture}
  \caption{The minimum DFA equivalent to $M_{A,0}$ in Example~\ref{example:csub}.}
  \label{fig:csub-r-min}
  \end{figure}
  %\par\medskip
  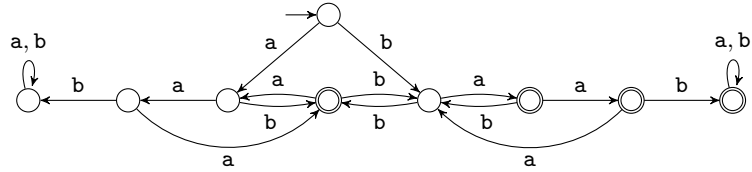
\begin{figure}[t]\centering
  \begin{tikzpicture}
    \node[initial,state] (q0) {};
    \node[state,accepting] (q0e) [below=0.8cm of q0] {};
    \node[state] (q1o) [left=of q0e] {};
    \node[state] (q2e) [left=of q1o] {};
    \node[state] (qx)  [left=of q2e] {};
    \node[state] (q0o) [right=of q0e] {};
    \node[state,accepting] (q1e) [right=of q0o] {};
    \node[state,accepting] (q2o) [right=of q1e] {};
    \node[state,accepting] (qa)  [right=of q2o] {};
    \path
    (q0)  edge node [swap,xshift=1mm] {$\t{a}$} (q1o)
          edge node [xshift=-1mm] {$\t{b}$} (q0o)
    (q0e) edge [bend right=10] node [swap] {$\t{a}$} (q1o)
          edge [bend left=10] node {$\t{b}$} (q0o)
    (q1o) edge node [swap] {$\t{a}$} (q2e)
          edge [bend right=10] node [swap,xshift=-1mm] {$\t{b}$} (q0e)
    (q2e) edge [out=-45,in=-135] node [swap] {$\t{a}$} (q0e)
          edge node [swap] {$\t{b}$} (qx)
    (qx)  edge [loop above] node {$\t{a},\t{b}$} (qx)
    (q0o) edge [bend left=10] node {$\t{a}$} (q1e)
          edge [bend left=10] node {$\t{b}$} (q0e)
    (q1e) edge node {$\t{a}$} (q2o)
          edge [bend left=10] node [xshift=1mm] {$\t{b}$} (q0o)
    (q2o) edge [out=-135,in=-45] node {$\t{a}$} (q0o)
          edge node {$\t{b}$} (qa)
    (qa)  edge [loop above] node {$\t{a},\t{b}$} (qa)
    ;
  \end{tikzpicture}
  \caption{A DFA recognizing $\Rev{L(M_{A,0})}$ in Example~\ref{example:csub}.}
  \label{fig:csubdfa}
  \end{figure}

  Figure~\ref{fig:csub-r-min} shows the minimum DFA
  equivalent to $M_{A,0}$, and
  Figure~\ref{fig:csubdfa} shows
  the minimum DFA recognizing $\Rev{L(M_{A,0})}$.
  %constructed from the DFA in Fig.~\ref{fig:csub-r-min}.
  In Example~7 of~\cite{IKIK26}, Inoue et al.\ have shown
  a necessary and sufficient condition on
  a word $w$ to satisfy $\Grd{G,A}(w)=0$ for this grammar~$G$:
  Let $w=\t{a}^{d_1}\t{b}^{e_1}\t{a}^{d_2}\t{b}^{e_2}\ldots
  \t{a}^{d_k}\t{b}^{e_k}$ with
  $d_2,\ldots,d_k,e_1,\ldots,e_{k-1} > 0$ and
  $d_1,e_k\ge 0$, and let
  $\{0,1,2\} \ni {d'_i} \equiv d_i \bmod 3$.
  Then, $\Grd{G,A}(w)=0$ if and only if
  (1) ${d'_i}\ne 2$ for each $1\le i\le k$ and
  $\sum_{i=1}^k {d'_i}+\sum_{i=1}^k e_i$ is even, or
  (2) there exists $1\le m\le k$ such that
  ${d'_1},\ldots,{d'_{m-1}}\ne 2$,
  ${d'_m}=2$, and
  $\sum_{i=1}^{m-1} {d'_i}+\sum_{i=1}^{m-1} e_i$ is odd.
  We can confirm that $\Rev{L(M_{A,0})}$ represented by
  the DFA in Fig.~\ref{fig:csubdfa} coincides with
  this condition.

\end{example}

%%%

\begin{example}
  \label{example:csub2}
  We modify the RLG in Example~\ref{example:csub} as follows:
  $G=(\{A,B\},\Sigma,R,A)$ where $\Sigma=\{\t{a},\t{b}\}$ and
  $R=\{\,
  A \to \t{a}A \mid \t{aa}A \mid B \mid \t{a}
  \mid \t{aa}$, $B \to \t{b}A \mid \t{b} \,\}$.

  Let $m_A=4$ and $m_B=2$.
  The construction of $M_{X,i}=(Q,\Sigma,\delta,q_0,F_{X,i})$
  for $X\in\{A,B\}$ and $i\le m_X$ is
  performed as follows.

  \medskip
  \noindent
  Step~1:\quad
  Let $\Rn=\{\, A\to \t{a}A \mid \t{aa}A$, $B\to \t{b}B \,\}$,
  $\Ru=\{A\to B\}$, and
  $\Rt'=\{\, A\to \t{a}\E \mid \t{aa}\E$, $B\to \t{b}\E \,\}$.
  Define $<_{\Gamma}$ as $B <_{\Gamma} A$.

  \smallskip
  \noindent
  Step~2:\quad
  Define $\Gamma'$ and $\GammaN$ in a similar way
  to Example~\ref{example:csub}.
  Let $q_0=\{\E^0\}$.

  \smallskip
  \noindent
  Steps~3 and~4:\quad
  As in Example~\ref{example:csub},
  we restrict $Q$ to the set of states
  reachable from~$q_0$ and
  define $\delta$ only on the restricted~$Q$.
  %Starting from $q_0$,
  %we iteratively apply Steps~3 and~4 to a newly found state~$q$.
  Since $B <_{\Gamma} A$, we define
  $J_{B,q,a}$ and $\delta'_B(q,a)$ before
  defining $J_{A,q,a}$
  for each $q$ and~$a$.
  We obtain the following
  $J_{X,q,a}$, $\delta'_X(q,a)$, and $\delta(q,a)$
  for $X\in\{A,B\}$, state $q$ reachable from $q_0$, and $a\in\{\t{a},\t{b}\}$.
  Note that
  $J_{B,q,\t{a}}=\delta'_B(q,\t{a})=\emptyset$
  for any $q$
  because the right-hand side of every rule for $B$ starts with~$\t{b}$.
\bgroup\small
\begin{align*}
  J_{A,         \{\E^0\},\t{a}} &= \{0\},
    & \delta'_A(\{\E^0\},\t{a}) &= \{A^1\},
    & \delta(   \{\E^0\},\t{a}) &= \{A^1,\t{a}\E^0\}. \\
  J_{B,         \{\E^0\},\t{b}} &= \{0\},
    & \delta'_B(\{\E^0\},\t{b}) &= \{B^1\}, \\
  J_{A,         \{\E^0\},\t{b}} &= \{1\},
    & \delta'_A(\{\E^0\},\t{b}) &= \{A^0\},
    & \delta(   \{\E^0\},\t{b}) &= \{A^0,B^1\}. \\
  J_{A,         \{A^1,\t{a}\E^0\},\t{a}} &= \{1,0\},
    & \delta'_A(\{A^1,\t{a}\E^0\},\t{a}) &= \{A^2\},
    & \delta(   \{A^1,\t{a}\E^0\},\t{a}) &= \{A^2,\t{a}A^1\}. \\
  J_{B,         \{A^1,\t{a}\E^0\},\t{b}} &= \{1\},
    & \delta'_B(\{A^1,\t{a}\E^0\},\t{b}) &= \{B^0\}, \\
  J_{A,         \{A^1,\t{a}\E^0\},\t{b}} &= \{0\},
    & \delta'_A(\{A^1,\t{a}\E^0\},\t{b}) &= \{A^1\},
    & \delta(   \{A^1,\t{a}\E^0\},\t{b}) &= \{A^1,B^0\}. \\
  J_{A,         \{A^0,B^1\},\t{a}} &= \{0\},
    & \delta'_A(\{A^0,B^1\},\t{a}) &= \{A^1\},
    & \delta(   \{A^0,B^1\},\t{a}) &= \{A^1,\t{a}A^0\}. \\
  J_{B,         \{A^0,B^1\},\t{b}} &= \{0\},
    & \delta'_B(\{A^0,B^1\},\t{b}) &= \{B^1\}, \\
  J_{A,         \{A^0,B^1\},\t{b}} &= \{1\},
    & \delta'_A(\{A^0,B^1\},\t{b}) &= \{A^0\},
    & \delta(   \{A^0,B^1\},\t{b}) &= \{A^0,B^1\}. \\
  J_{A,         \{A^2,\t{a}A^1\},\t{a}} &= \{2,1\},
    & \delta'_A(\{A^2,\t{a}A^1\},\t{a}) &= \{A^0\},
    & \delta(   \{A^2,\t{a}A^1\},\t{a}) &= \{A^0,\t{a}A^2\}. \\
  J_{B,         \{A^2,\t{a}A^1\},\t{b}} &= \{2\},
    & \delta'_B(\{A^2,\t{a}A^1\},\t{b}) &= \{B^0\}, \\
  J_{A,         \{A^2,\t{a}A^1\},\t{b}} &= \{0\},
    & \delta'_A(\{A^2,\t{a}A^1\},\t{b}) &= \{A^1\},
    & \delta(   \{A^2,\t{a}A^1\},\t{b}) &= \{A^1,B^0\}. \\
  J_{A,         \{A^1,B^0\},\t{a}} &= \{1\},
    & \delta'_A(\{A^1,B^0\},\t{a}) &= \{A^0\},
    & \delta(   \{A^1,B^0\},\t{a}) &= \{A^0,\t{a}A^1\}. \\
  J_{B,         \{A^1,B^0\},\t{b}} &= \{1\},
    & \delta'_B(\{A^1,B^0\},\t{b}) &= \{B^0\}, \\
  J_{A,         \{A^1,B^0\},\t{b}} &= \{0\},
    & \delta'_A(\{A^1,B^0\},\t{b}) &= \{A^1\},
    & \delta(   \{A^1,B^0\},\t{b}) &= \{A^1,B^0\}. \\
  J_{A,         \{A^1,\t{a}A^0\},\t{a}} &= \{1,0\},
    & \delta'_A(\{A^1,\t{a}A^0\},\t{a}) &= \{A^2\},
    & \delta(   \{A^1,\t{a}A^0\},\t{a}) &= \{A^2,\t{a}A^1\}. \\
  J_{B,         \{A^1,\t{a}A^0\},\t{b}} &= \{1\},
    & \delta'_B(\{A^1,\t{a}A^0\},\t{b}) &= \{B^0\}, \\
  J_{A,         \{A^1,\t{a}A^0\},\t{b}} &= \{0\},
    & \delta'_A(\{A^1,\t{a}A^0\},\t{b}) &= \{A^1\},
    & \delta(   \{A^1,\t{a}A^0\},\t{b}) &= \{A^1,B^0\}. \\
  J_{A,         \{A^0,\t{a}A^2\},\t{a}} &= \{0,2\},
    & \delta'_A(\{A^0,\t{a}A^2\},\t{a}) &= \{A^1\},
    & \delta(   \{A^0,\t{a}A^2\},\t{a}) &= \{A^1,\t{a}A^0\}. \\
  J_{B,         \{A^0,\t{a}A^2\},\t{b}} &= \{0\},
    & \delta'_B(\{A^0,\t{a}A^2\},\t{b}) &= \{B^1\}, \\
  J_{A,         \{A^0,\t{a}A^2\},\t{b}} &= \{1\},
    & \delta'_A(\{A^0,\t{a}A^2\},\t{b}) &= \{A^0\},
    & \delta(   \{A^0,\t{a}A^2\},\t{b}) &= \{A^0,B^1\}. \\
  J_{A,         \{A^0,\t{a}A^1\},\t{a}} &= \{0,1\},
    & \delta'_A(\{A^0,\t{a}A^1\},\t{a}) &= \{A^2\},
    & \delta(   \{A^0,\t{a}A^1\},\t{a}) &= \{A^2,\t{a}A^0\}. \\
  J_{B,         \{A^0,\t{a}A^1\},\t{b}} &= \{0\},
    & \delta'_B(\{A^0,\t{a}A^1\},\t{b}) &= \{B^1\}, \\
  J_{A,         \{A^0,\t{a}A^1\},\t{b}} &= \{1\},
    & \delta'_A(\{A^0,\t{a}A^1\},\t{b}) &= \{A^0\},
    & \delta(   \{A^0,\t{a}A^1\},\t{b}) &= \{A^0,B^1\}. \\
  J_{A,         \{A^2,\t{a}A^0\},\t{a}} &= \{2,0\},
    & \delta'_A(\{A^2,\t{a}A^0\},\t{a}) &= \{A^1\},
    & \delta(   \{A^2,\t{a}A^0\},\t{a}) &= \{A^1,\t{a}A^2\}. \\
  J_{B,         \{A^2,\t{a}A^0\},\t{b}} &= \{2\},
    & \delta'_B(\{A^2,\t{a}A^0\},\t{b}) &= \{B^0\}, \\
  J_{A,         \{A^2,\t{a}A^0\},\t{b}} &= \{0\},
    & \delta'_A(\{A^2,\t{a}A^0\},\t{b}) &= \{A^1\},
    & \delta(   \{A^2,\t{a}A^0\},\t{b}) &= \{A^1,B^0\}. \\
  J_{A,         \{A^1,\t{a}A^2\},\t{a}} &= \{1,2\},
    & \delta'_A(\{A^1,\t{a}A^2\},\t{a}) &= \{A^0\},
    & \delta(   \{A^1,\t{a}A^2\},\t{a}) &= \{A^0,\t{a}A^1\}. \\
  J_{B,         \{A^1,\t{a}A^2\},\t{b}} &= \{1\},
    & \delta'_B(\{A^1,\t{a}A^2\},\t{b}) &= \{B^0\}, \\
  J_{A,         \{A^1,\t{a}A^2\},\t{b}} &= \{0\},
    & \delta'_A(\{A^1,\t{a}A^2\},\t{b}) &= \{A^1\},
    & \delta(   \{A^1,\t{a}A^2\},\t{b}) &= \{A^1,B^0\}. \\
\end{align*}
\egroup

  \noindent
  Step~5:\quad
  Define $F_{A,i}$ and $F_{B,j}$
  for $0\le i\le 4$ and $0\le j\le 2$
  in a similar way to Example~\ref{example:csub}.
  Note that $F_{A,3}=F_{A,4}=F_{B,2}=\emptyset$.
  \medskip

  \begin{figure}[t]\centering
  \begin{tikzpicture}[scale=0.9,transform shape]
    \node[initial,elliptic state] (E0) {$\E^0$};
    \node[elliptic state] (A1aE0) [right=of E0] {$A^1,\t{a}\E^0$};
    \node[elliptic state] (A2aA1) [right=of A1aE0] {$A^2,\t{a}A^1$};
    \node[elliptic state] (A0aA2) [right=of A2aA1] {$A^0,\t{a}A^2$};
    \node[elliptic state] (A0B1) [below=of A1aE0] {$A^0,B^1$};
    \node[elliptic state] (A1B0) [below=of A2aA1] {$A^1,B^0$};
    \node[elliptic state] (A1aA0) [below=of A0aA2] {$A^1,\t{a}A^0$};
    \node[elliptic state] (A0aA1) [below=of A0B1] {$A^0,\t{a}A^1$};
    \node[elliptic state] (A2aA0) [below=of A1B0] {$A^2,\t{a}A^0$};
    \node[elliptic state] (A1aA2) [below=of A1aA0] {$A^1,\t{a}A^2$};
    \path
    (E0)    edge node {$\t{a}$} (A1aE0)
            edge [bend right] node {$\t{b}$} (A0B1)
    (A1aE0) edge node {$\t{a}$} (A2aA1)
            edge node [yshift=-3mm,xshift=4mm] {$\t{b}$} (A1B0)
    (A2aA1) edge node {$\t{a}$} (A0aA2)
            edge node [yshift=-2mm] {$\t{b}$} (A1B0)
    (A0aA2) edge node {$\t{a}$} (A1aA0)
            edge [out=-150,in=45] node [swap,xshift=24mm,yshift=-2mm] {$\t{b}$} (A0B1)
    (A0B1)  edge [bend right=20] node [xshift=-18mm,yshift=-1mm] {$\t{a}$} (A1aA0)
            edge [loop above] node {$\t{b}$} (A0B1)
    (A1B0)  edge node {$\t{a}$} (A0aA1)
            edge [loop left,out=-170,in=170,looseness=7] node {$\t{b}$} (A1B0)
    (A1aA0) edge node [swap,yshift=-3mm,xshift=4mm] {$\t{a}$} (A2aA1)
            edge node [swap] {$\t{b}$} (A1B0)
    (A0aA1) edge node {$\t{a}$} (A2aA0)
            edge node {$\t{b}$} (A0B1)
    (A2aA0) edge node {$\t{a}$} (A1aA2)
            edge node [yshift=-2mm] {$\t{b}$} (A1B0)
    (A1aA2) edge [bend left=25] node {$\t{a}$} (A0aA1)
            edge node {$\t{b}$} (A1B0)
    ;
  \end{tikzpicture}
  \caption{DFA $(Q,\Sigma,\delta,q_0,\emptyset)$
    in Example~\ref{example:csub2}.}
  \label{fig:csub2-r}
  \end{figure}
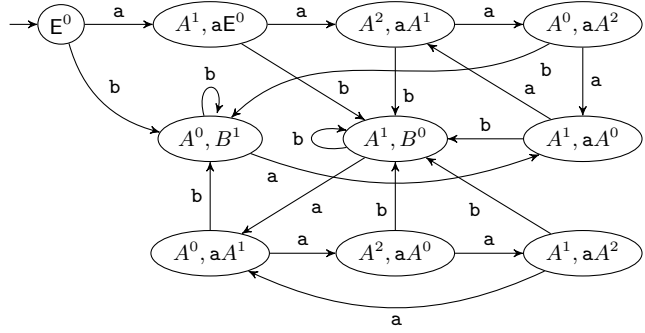

  Figure~\ref{fig:csub2-r} shows the constructed DFA
  $(Q,\Sigma,\delta,q_0,\emptyset)$ without final states.
  By simulating the behavior of this DFA
  on the words of length 4 or less,
  we obtain the following table of languages.
  %We can confirm that every word $w$ appearing in this table
  %satisfies
  %``$w\in L(M_{A,i}) \iff \Grd{G,A}(\Rev{w})=i$.''
  %
  \begin{align*}
    L(M_{A,0})&=\{
      \t{b},\t{bb},\t{aaa},\t{aba},\t{bbb},
      \t{aaab},\t{aaba},\t{abab},\t{abba},\t{baaa},\t{baba},\t{bbbb},
      \ldots
    \}, \\
    L(M_{A,1})&=\{
      \begin{aligned}[t]
      &\t{a},\t{ab},\t{ba}, \t{aab},\t{abb},\t{bab},\t{bba},
      \t{aaaa},\t{aabb},\t{abbb},\t{baab},\t{babb},\t{bbab},
      \\ &\t{bbba},
      \ldots
    \}, \end{aligned} \\
    L(M_{A,2})&=\{
      \t{aa},\t{baa},\t{abaa},\t{bbaa},\ldots
    \}, \\
    L(M_{B,0})&=\{
      \begin{aligned}[t]
      & \t{ab},\t{aab},\t{abb},\t{bab},
        \t{aabb},\t{abbb},\t{baab},\t{babb},\t{bbab},
      \ldots
    \}, \end{aligned} \\
    L(M_{B,1})&=\{
      \t{b},\t{bb}, \t{bbb},
      \t{aaab},\t{abab},\t{bbbb},
      \ldots
    \}, \\
    L(M_{A,3})&=L(M_{A,4})=L(M_{B,2})=\emptyset.
  \end{align*}

  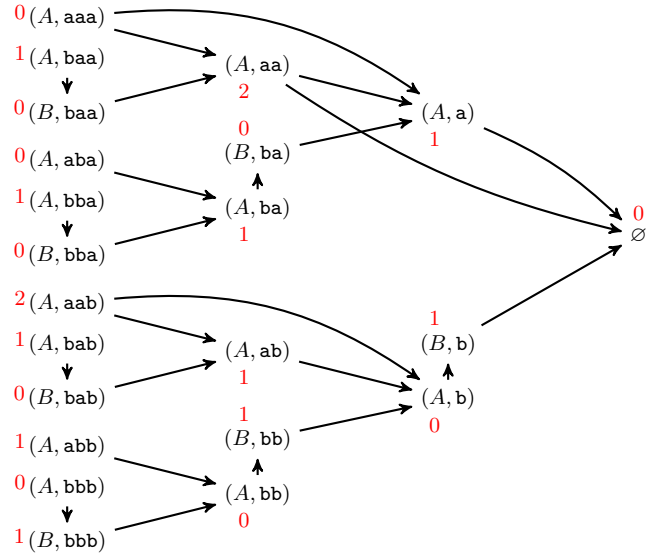
\begin{figure} \centering
  \begin{tikzpicture}[->, thick, scale=0.9, transform shape]
  \tikzstyle{every node}=[font=\small]
  \foreach \w [count=\i from 0] in {aa, ba, ab, bb}
  {
    \node (Aa\w) at (0, -2.1 * \i) {$(A,\t{a\w})$};
    \node [below=0.6 of Aa\w.center, anchor=center] (Ab\w) {$(A,\t{b\w})$};
    \node [below=0.8 of Ab\w.center, anchor=center] (Bb\w) {$(B,\t{b\w})$};
    \draw (Ab\w) -- (Bb\w);
  }
  \foreach \w [count=\i from 0] in {ba,aa,ab,bb}
  {
    \node[below right=0.7 and 2.8 of Aa\w.center, anchor=center] (A\w) {$(A,\t{\w})$};
    \draw (Aa\w) -- (A\w);
    \draw (Bb\w) -- (A\w);
  }
  \foreach \w [count=\i from 0] in {ba,bb}
  {
    \node[above=0.8 of A\w.center, anchor=center] (B\w) {$(B,\t{\w})$};
    \draw (A\w) -- (B\w);
  }
  \foreach \w [count=\i from 0] in {a,b}
  {
    \node[below right=0.7 and 2.8 of Aa\w.center, anchor=center] (A\w) {$(A,\t{\w})$};
    \draw (Aa\w) -- (A\w);
    \draw (Bb\w) -- (A\w);
  }
  \node[above=0.8 of Ab.center, anchor=center] (Bb) {$(B,\t{b})$};
  \draw (Ab) -- (Bb);
  \node[below right=1.8 and 2.8 of Aa.center, anchor=center] (F) {$\varnothing$};
  \draw (Aa) to [bend left=10] (F);
  \draw (Bb) -- (F);
  \draw (Aaaa) to [bend left=20] (Aa);
  \draw (Aaab) to [bend left=20] (Ab);
  \draw (Aaa) to [bend right=10] (F);
  \tikzstyle{every node}=[font=\footnotesize, text=red]
  \foreach \w / \g in {Aaaa/0, Abaa/1, Bbaa/0, Aaba/0, Abba/1, Bbba/0,
                       Aaab/2, Abab/1, Bbab/0, Aabb/1, Abbb/0, Bbbb/1}
  {
    \node[below=0.2 of \w.north west, anchor=center] {$\g$};
  }
  \foreach \w / \g in {aa/2, ba/1, ab/1, bb/0, a/1, b/0}
  {
    \node[below left=0.1 and 0.2 of A\w.south, anchor=center] {$\g$};
  }
  \foreach \w / \g in {ba/0, bb/1, b/1}
  {
    \node[above left=0.1 and 0.2 of B\w.north, anchor=center] {$\g$};
  }
  \node[above=0.1 of F.north, anchor=center] {$0$};
  \end{tikzpicture}
  \caption{A partial game graph of the successful derivation game
    in Example~\ref{example:csub2}.}
  \label{fig:graph}
  \end{figure}%
  Figure~\ref{fig:graph} shows a partial game graph
  of the successful derivation game on~$G$,
  %In the graph, each singleton position $\{(X,w)\}$
  %is written as $(X,w)$ for simplicity.
  which contains every position $(X,w)$ such that $|w|\le 3$
  and also shows the Grundy number of each position.
  Note that $\{(B,\t{a}x)\}$ for $x\in\Sigma^*$ is not
  a position of this game because $L(G,B)=\t{b}\Sigma^*$.
  We can see that the Grundy numbers in Fig.~\ref{fig:graph}
  satisfy ``$\Rev{w}\in L(M_{X,i}) \iff \Grd{G,X}(w)=i$.''
    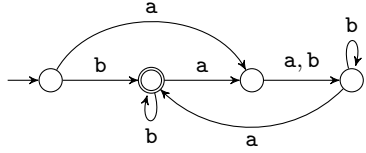
\begin{figure}\centering
  \begin{tikzpicture}
    \node[initial,state] (E0) {};
    \node[state,accepting] (q0) [right=of E0] {};
    \node[state] (q1) [right=of q0] {};
    \node[state] (q2) [right=of q1] {};
    \path
    (E0) edge [bend left=60] node {$\t{a}$} (q1)
         edge node {$\t{b}$} (q0)
    (q0) edge node {$\t{a}$} (q1)
         edge [loop below] node {$\t{b}$} (q0)
    (q1) edge node {$\t{a},\t{b}$} (q2)
    (q2) edge [bend left=45] node {$\t{a}$} (q0)
         edge [loop above] node {$\t{b}$} (q2)
    ;
  \end{tikzpicture}
  \caption{The minimum DFA equivalent to $M_{A,0}$ in Example~\ref{example:csub2}.}
  \label{fig:csub2-r-min}
  \end{figure}
  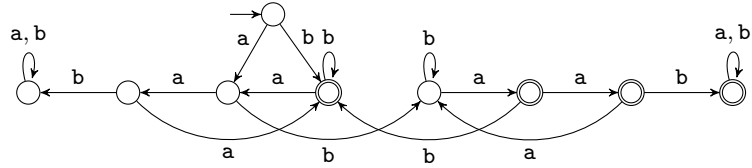
\begin{figure}\centering
  \begin{tikzpicture}
    \node[initial,state] (q0) {};
    \node[state,accepting] (q0e) [below right=0.8cm and 0.5cm of q0] {};
    \node[state] (q1o) [left=of q0e] {};
    \node[state] (q2e) [left=of q1o] {};
    \node[state] (qx)  [left=of q2e] {};
    \node[state] (q0o) [right=of q0e] {};
    \node[state,accepting] (q1e) [right=of q0o] {};
    \node[state,accepting] (q2o) [right=of q1e] {};
    \node[state,accepting] (qa)  [right=of q2o] {};
    \path
    (q0)  edge node [swap,xshift=1mm] {$\t{a}$} (q1o)
          edge node [xshift=-1mm] {$\t{b}$} (q0e)
    (q0e) edge node [swap] {$\t{a}$} (q1o)
          edge [loop above] node {$\t{b}$} (q0e)
    (q1o) edge node [swap] {$\t{a}$} (q2e)
          edge [out=-45,in=-135] node [swap] {$\t{b}$} (q0o)
    (q2e) edge [out=-45,in=-135] node [swap] {$\t{a}$} (q0e)
          edge node [swap] {$\t{b}$} (qx)
    (qx)  edge [loop above] node {$\t{a},\t{b}$} (qx)
    (q0o) edge node {$\t{a}$} (q1e)
          edge [loop above] node {$\t{b}$} (q0o)
    (q1e) edge node {$\t{a}$} (q2o)
          edge [out=-135,in=-45] node {$\t{b}$} (q0e)
    (q2o) edge [out=-135,in=-45] node {$\t{a}$} (q0o)
          edge node {$\t{b}$} (qa)
    (qa)  edge [loop above] node {$\t{a},\t{b}$} (qa)
    ;
  \end{tikzpicture}
  \caption{A DFA recognizing $\Rev{L(M_{A,0})}$ in Example~\ref{example:csub2}.}
  \label{fig:csub2dfa}
  \end{figure}%

  Figure~\ref{fig:csub2-r-min} shows the minimum DFA
  equivalent to $M_{A,0}$, and
  Figure~\ref{fig:csub2dfa} shows
  the minimum DFA recognizing $\Rev{L(M_{A,0})}$.
\end{example}

%%% ----- Examples -----

\section{Conclusion}
\label{sec:conclusion}

In this paper, we studied
the successful derivation game (SDG)
on a right-linear grammar (RLG)
and showed that
we can construct a DFA for computing the Grundy number of
a given position.
This means that for this game,
the set of positions with a given Grundy number~$c$
is regular.
As a corollary, we showed that
the least upper bound of the Grundy numbers
in the SDG on a given RLG is decidable,
which is known to be undecidable in general
for the SDG on a linear CFG\@.
%As another corollary,
%both
%$L_{\cal{P}}=\{w\in\Sigma^+\mid \Grd{G,I}(w)=0\}$
%(the set of initial positions where the first player loses)
%and
%$L_{\cal{N}}=\{w\in\Sigma^+\mid \Grd{G,I}(w)\ne 0\}$
%(the set of initial positions where the first player wins)
%are regular.
%Future work includes investigating
%the complexity of
We also showed that
computing the least upper bound of the Grundy numbers
in the SDG on an RLG is PSPACE-complete.

%%%%%

\bibliographystyle{splncs04}
\bibliography{ncrsp}

\end{document}